\documentclass[11pt,a4paper]{article}

\usepackage{fontspec}
\usepackage{amsmath,amssymb,amsthm}
\usepackage[margin=1in]{geometry}
\usepackage{hyperref}
\usepackage{graphicx}
\usepackage{tikz}
\usetikzlibrary{arrows.meta,positioning,fit,backgrounds}

\usepackage{polyglossia}
\setmainlanguage{english}
\setotherlanguage{arabic}
\newfontfamily\arabicfont[Script=Arabic]{Amiri-Regular.ttf}

\newtheorem{theorem}{Theorem}
\newtheorem{proposition}{Proposition}
\newtheorem{lemma}{Lemma}
\newtheorem{corollary}{Corollary}
\theoremstyle{definition}
\newtheorem{definition}{Definition}
\newtheorem{remark}{Remark}

\title{A Model-Theoretic Formalization of Al-Khw\={a}rizm\={i}'s
\textit{'ilm al-far\={a}'i\={d}}}

\author{Elnaserledinellah Mahmoud Abdelwahab\\
\href{mailto:elnaser@computational-fiqh.org}
{\texttt{elnaser@computational-fiqh.org}}}

\date{}

\begin{document}
\maketitle

\begin{abstract}
This paper presents a model-theoretic and computational formalization of
Al-Khw\={a}rizm\={i}'s treatment of inheritance (\textit{`ilm
al-far\={a}'i\d{d}}), within the framework of Computational Fiqh. We
first show that every inheritance configuration in a specified class is
representable as a finite system of rational linear equalities and
inequalities, and that such systems are decided exactly by
Fourier--Motzkin elimination over $\mathbb{Q}$. We then define the
fixed-dimensional inheritance structure $\mathcal{V}_d=\mathbb{Q}^d$
with distinguished claimant coordinates, and prove that every admissible
rational linear-coordinate realization is isomorphic to $\mathcal{V}_d$;
this establishes the algebraic representation as \emph{canonical} within
that class. The complete first-order theory $T_d$ of $\mathcal{V}_d$ is
shown to be complete by construction---a design property of the
formalization, not a novel theorem---and is explicitly distinguished
from the decidability of the inheritance-relevant fragment. The
methodological principle is then related to a broader formalization of
closed Fiqh domains via exhaustive case generation, Horn-clause
axiomatization, and terminating deduction, yielding
\emph{domain-relative computational completeness}.
\end{abstract}

\section{Introduction}

Al-Khw\={a}rizm\={i}'s treatment of inheritance
(\textit{'ilm al-far\={a}'i\={d}}) is significant not merely because it
provides solutions to individual numerical problems, but because it
provides general computational procedures for specified classes of
inheritance configurations. The underlying methodological principle is
the complete computational treatment of a delimited jurisprudential
domain.

The purpose of this paper is to give that principle a modern formal
interpretation. A specified inheritance problem is represented by a
finite system of linear equations and inequalities over the rational
numbers. For a fixed inheritance configuration, the resulting
mathematical structure can be represented by a finite-dimensional
rational coordinate space. This permits a model-theoretic
characterization of the intended structure together with a separate
algorithmic analysis of the corresponding constraint systems.

For a long period, Western ``Islamic Studies'' provided the principal
academic setting in which questions concerning Fiqh, its branches, and
its related sciences were investigated within modern academia. The
methodological traditions that developed within this field, however,
were not designed primarily for the analysis of formally specified
mathematical, logical, or computational systems. This distinction
becomes consequential when some branches of Fiqh, such as the
jurisprudence of inheritance, are understood to contain precisely such
systems: they involve defined quantities, algebraic relations,
dependencies, exhaustive cases, and procedures for deriving
determinate results.

The resulting epistemological clash is a central concern of this
paper. A number of modern discussions within Islamic Studies have
advanced claims that particular Fiqh rules or systems are legally,
mathematically, or logically inconsistent. Such claims cannot be
established merely by producing a contradiction under an informal
interpretation of selected texts or rules. In particular, a fundamental
distinction must be maintained between linguistic ambiguity and formal
coherence. The existence of several linguistically possible
interpretations does not by itself imply that the underlying formal
system is inconsistent. If a textual formulation admits several
candidate interpretations
\[
I_1,I_2,\ldots,I_k,
\]
then a contradiction obtained under one interpretation establishes, at
most, the inconsistency of the corresponding formalization; it does not
establish the inconsistency of the underlying jurisprudential system
unless that interpretation has first been justified as the relevant
one. Conversely, the existence of an interpretation under which the
resulting formal system is coherent is mathematically significant even
when other interpretations are inconsistent. This distinction is
essential whenever linguistic material is used as the basis for claims
about mathematical or logical properties.

Section~2 examines this methodological problem directly. It considers
claims concerning legal, mathematical, and logical inconsistency and
the conceptual assumptions underlying them, including the failure to
distinguish linguistic ambiguity from formal inconsistency. The section
argues that conclusions about the consistency or inconsistency of a
jurisprudential computational system require an explicit formal
reconstruction and analysis of that system. The present paper then
provides such a reconstruction and uses the methods of Computational
Fiqh to determine its mathematical and logical properties explicitly.

Computational Fiqh emerges in this context not merely as the application
of computational tools to an already established body of
jurisprudential scholarship, but as a distinct methodological framework
for identifying, formalizing, analyzing, verifying, and implementing
computational structures occurring within delimited Fiqh domains. Its
epistemological significance is therefore also critical: where a
jurisprudential domain gives rise to a formally specifiable system,
claims about that system should be settled by rigorous formal analysis
rather than by methodologies incapable of representing the system
adequately. The present study uses the inheritance system as a concrete
case in which this distinction can be demonstrated. This contrast is visualized in
Figure~\ref{fig:epistemological-distinction}.


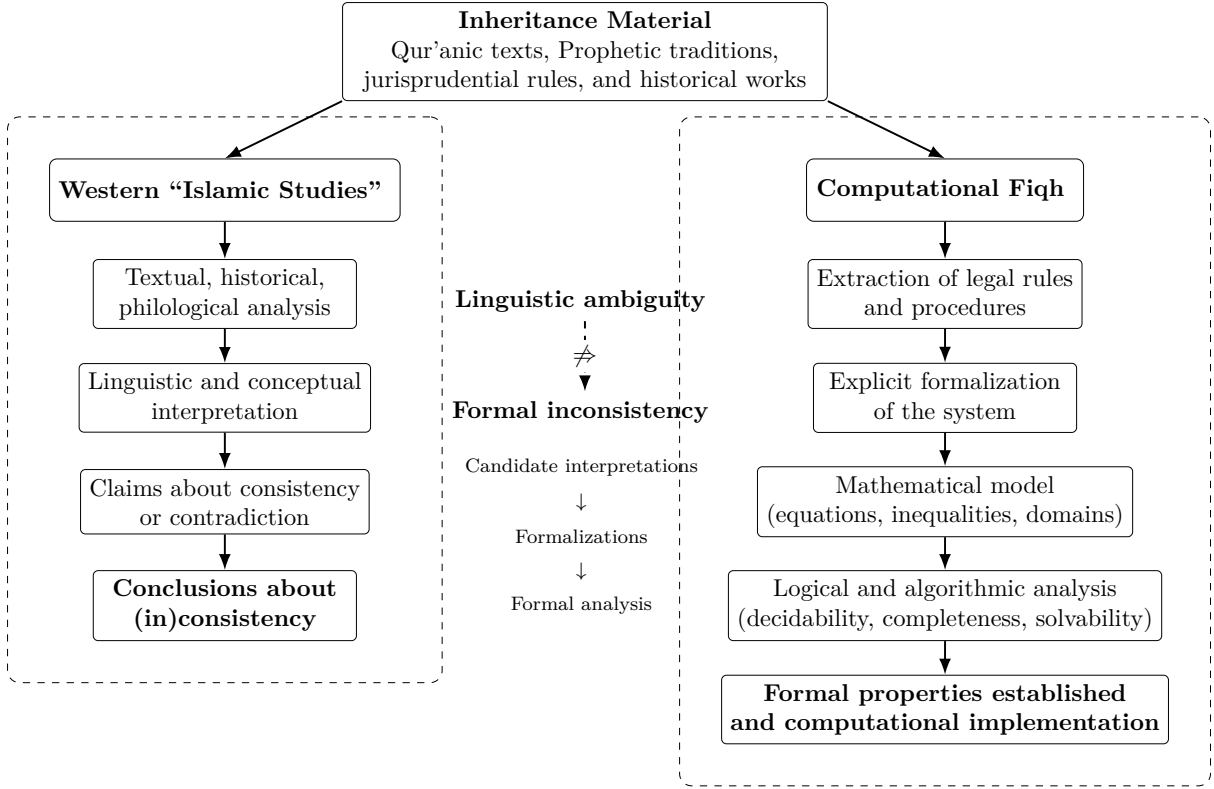
\begin{figure}[t]
\centering
\resizebox{\textwidth}{!}{%
\begin{tikzpicture}[
    >=Latex,
    box/.style={
        draw,
        rounded corners=2pt,
        align=center,
        minimum width=38mm,
        minimum height=8mm,
        font=\small
    },
    titlebox/.style={
        draw,
        rounded corners=3pt,
        align=center,
        minimum width=40mm,
        minimum height=9mm,
        font=\bfseries\small
    },
    arrow/.style={
        -Latex,
        thick
    }
]


\node[
    box,
    minimum width=70mm
] (material) at (0,0) {
    \textbf{Inheritance Material}\\
    Qur'anic texts, Prophetic traditions,\\
    jurisprudential rules, and historical works
};


\node[
    titlebox
] (is-title) at (-52mm,-20mm) {
    Western ``Islamic Studies''
};

\node[
    titlebox
] (cf-title) at (52mm,-20mm) {
    Computational Fiqh
};


\node[
    box
] (is-analysis) at (-52mm,-35mm) {
    Textual, historical,\\
    philological analysis
};

\node[
    box
] (is-interpretation) at (-52mm,-50mm) {
    Linguistic and conceptual\\
    interpretation
};

\node[
    box
] (is-claims) at (-52mm,-65mm) {
    Claims about consistency\\
    or contradiction
};

\node[
    box
] (is-conclusion) at (-52mm,-80mm) {
    \textbf{Conclusions about}\\
    \textbf{(in)consistency}
};


\node[
    box
] (cf-extraction) at (52mm,-35mm) {
    Extraction of legal rules\\
    and procedures
};

\node[
    box
] (cf-formal) at (52mm,-50mm) {
    Explicit formalization\\
    of the system
};

\node[
    box
] (cf-model) at (52mm,-65mm) {
    Mathematical model\\
    (equations, inequalities, domains)
};

\node[
    box
] (cf-analysis) at (52mm,-80mm) {
    Logical and algorithmic analysis\\
    (decidability, completeness, solvability)
};

\node[
    box
] (cf-result) at (52mm,-95mm) {
    \textbf{Formal properties established}\\
    \textbf{and computational implementation}
};


\draw[arrow]
    (material.south west) -- (is-title.north);

\draw[arrow]
    (material.south east) -- (cf-title.north);


\draw[arrow]
    (is-title) -- (is-analysis);

\draw[arrow]
    (is-analysis) -- (is-interpretation);

\draw[arrow]
    (is-interpretation) -- (is-claims);

\draw[arrow]
    (is-claims) -- (is-conclusion);


\draw[arrow]
    (cf-title) -- (cf-extraction);

\draw[arrow]
    (cf-extraction) -- (cf-formal);

\draw[arrow]
    (cf-formal) -- (cf-model);

\draw[arrow]
    (cf-model) -- (cf-analysis);

\draw[arrow]
    (cf-analysis) -- (cf-result);


\node[
    align=center,
    font=\small\bfseries
] (ambiguity) at (0,-36mm) {
    Linguistic ambiguity
};

\node[
    align=center,
    font=\small\bfseries
] (inconsistency) at (0,-52mm) {
    Formal inconsistency
};

\draw[
    -Latex,
    thick,
    dashed
] (ambiguity) -- node[
    fill=white,
    inner sep=2pt,
    font=\small\bfseries
] {$\not\Rightarrow$} (inconsistency);

\node[
    align=center,
    font=\scriptsize
] at (0,-60mm) {
    Candidate interpretations
};

\node[
    align=center,
    font=\scriptsize
] at (0,-65mm) {
    $\downarrow$
};

\node[
    align=center,
    font=\scriptsize
] at (0,-70mm) {
    Formalizations
};

\node[
    align=center,
    font=\scriptsize
] at (0,-75mm) {
    $\downarrow$
};

\node[
    align=center,
    font=\scriptsize
] at (0,-80mm) {
    Formal analysis
};


\begin{scope}[on background layer]

    \node[
        draw,
        rounded corners=5pt,
        dashed,
        inner sep=6mm,
        fit=(is-title)
            (is-analysis)
            (is-interpretation)
            (is-claims)
            (is-conclusion)
    ] {};

    \node[
        draw,
        rounded corners=5pt,
        dashed,
        inner sep=6mm,
        fit=(cf-title)
            (cf-extraction)
            (cf-formal)
            (cf-model)
            (cf-analysis)
            (cf-result)
    ] {};

\end{scope}

\end{tikzpicture}%
}

\caption{Epistemological distinction between Western "Islamic Studies" and Computational Fiqh in the treatment of the Inheritance System. Linguistic
ambiguity does not by itself entail formal inconsistency as claimed in "Islamic Studies". Computational
Fiqh makes the formal system explicit before mathematical and logical
properties are inferred.}
\label{fig:epistemological-distinction}
\end{figure}

Three formal properties are considered. First, the admissible
fixed-dimensional inheritance structure is canonical within its class
up to isomorphism. Second, its complete first-order theory is complete
by construction---a design property of the formalization rather than a
novel theorem. Third, the finite systems of rational linear constraints
arising from individual inheritance configurations are algorithmically
decidable.

These properties are distinct. Completeness of $T_d$ concerns the
first-order theory of the intended mathematical structure; algorithmic
decidability concerns the existential fragment generated by finite
constraint systems. Neither notion should be conflated with the
historical question of whether the original computational procedures
were expressed in the modern formal language used here.

The final part relates this construction to a broader formalization of
closed Fiqh domains based on exhaustive case generation, Horn-clause
axiomatization, and terminating deduction. This motivates the notion
of \emph{domain-relative computational completeness}: the computational
completeness of a formal system relative to a specified and delimited
domain, rather than an unrestricted claim about jurisprudential
reasoning as a whole.

\subsection{Contributions and non-claims}

\paragraph{Contributions.}
\begin{enumerate}
\item A \emph{representation lemma} (\S\ref{sec:rep}): every
inheritance configuration in the specified class is a finite rational
linear constraint system.
\item A \emph{canonicity theorem} (\S\ref{sec:canon}): every admissible
rational linear-coordinate model is isomorphic to $\mathbb{Q}^d$.
\item A \emph{decision procedure} for the inheritance-relevant fragment
(\S\ref{sec:dec}): exact, terminating, rational.
\item A \emph{methodological interpretation} for Computational Fiqh
(\S\ref{sec:fiqh}).
\end{enumerate}

\paragraph{Non-claims.}
\begin{enumerate}
\item Full decidability of the first-order theory $T_d$ is \emph{not}
established. The language $\mathcal{L}_d$ is uncountable and the paper
decides only the existential fragment over finite quantifier-free
linear constraints.
\item \emph{Methodological exclusivity} of linear algebra is \emph{not}
claimed. The claim is representational canonicity. As noted in
\S2.3, \textit{`awl} resolves $\sum s_i X=1$ and is algebra in
disguise.
\item No historical claim is made that Al-Khw\={a}rizm\={i} used
Fourier--Motzkin; the procedure is a modern formal reconstruction
(\S\ref{sec:worked}, final paragraph).
\item Domain-relative computational completeness is conditional on the
results of the companion framework cited in \S\ref{sec:fiqh}.
\end{enumerate}


\section{The Quranic origin of Islamic Inheritance Systems}

Contrary to other known inheritance systems, the fundamental normative
rules from which the formal axioms of the Islamic inheritance system can
be extracted are directly revealed in the Quran across three specific
verses. All subsequent legal and mathematical developments, including
the algebraic framework formalized by Al-Khw\={a}rizm\={i}, were
fundamentally interpretations and computational executions of these core
scriptural axioms \cite{Rashed1994, Berggren1986}.

While Western ``Islamic studies'' disciplines, which were historically
conceived under orientalist paradigms for colonial purposes
\cite{Said1978, Tibawi1964} are subjected since then to substantial
revisions in some important areas of investigation, focusing research
more on objective evidence than speculation, they still routinely refuse
to accept the sudden, unprecedented emergence of the Quranic revelation,
attributing its origins instead to a variety of unsupported factors. This stance has become an unfalsifiable doctrine-like hypothesis questioning many results reported to be scientific.

This historical bias might have remained unnoticed by the present work belonging to the emerging new field of
``Computational Islamic Jurisprudence (Fiqh)'', which is neither
theological nor historical, but normative and epistemologically similar
to empirical sciences, were it not for the fact that some historicist
speculations have overreached their qualitative boundaries to make
flawed assertions regarding both the mathematical correctness
\cite{Crone1987, Burton1990} and the legal coherence \cite{Powers2011}
of the underlying Quranic text. Such an analytical overreach is
reminiscent of the 1970s revisionist crisis, where the field's speculative theories regarding an alleged late
Abbasid origin of the Quran collapsed the moment they were confronted
with the hard, empirical boundary conditions of radiocarbon dating
\cite{Deroche2014, Sadeghi2012}.

Consequently, this section is epistemologically necessary from the
perspective of Computational Fiqh, whose core mandate is to analyze,
formalize and implement Fiqh chapters using rigorous
mathematical-logical methods, to address the structural limitations and
explanatory overreach found in these recent academic treatments of the
Quranic inheritance system.

While the "Islamic Studies" field often claims alignment with the rigor of
``critical rationalism'', the reliance on unobservable internal
motivations, retrospective narratives, and unfalsifiable interpretive
frameworks prevents the majority of its findings from qualifying as
genuine scientific findings under Karl Popper's own demarcation criteria
\cite{Popper1963, Hallaq2013}. Despite this lack of empirical
falsifiability, some practitioners frequently overextend their
qualitative tools to evaluate complex, internal mathematical structures
that lie entirely outside their discipline's scope, as well as their
linguistic and mathematical capacities. This boundary failure is vividly
illustrated when academics analyze the three foundational inheritance
verses of Sura Al-Nisa.

For instance, when confronting the Kalalah scenarios in verses 12 and
176, David Powers assumes that the term ``siblings'' \textarabic{إخوة}
(ikhwah) in both verses refers to the exact same biological relation. He
premises this on the claim that the syntax provides no explicit
differentiation between the two cases, leading him to declare the
underlying text legally incoherent and inherently contradictory and that
later scholars have ``fixed'' this incoherence. This assertion, however,
collapses under legal, linguistic, and formal logical scrutiny.

First, the claim that the inheritance mechanism, when completely
stripped from its context, is ``legally incoherent'' refers to an
unspecified notion of ``legality'', when it ignores both the
foundational principles of the Islamic Legal System and the nature of
Quranic revelation itself. The Quran was revealed in an event-based,
incremental manner \textarabic{النزول المنجم} (nuzul manjam), whereby
each verse has what classical Muslim scholars term Asbab al-Nuzul
(\textarabic{أسباب النزول} reasons of revelation)--- a context largely
determined by the actions and reactions of the Prophet (pbuh) and his
companions to real-world events. Verse 176 explicitly opens with
\textarabic{يستفتونك} ``yastaftunaka'' (``they ask you''), proving there
was a specific question-context to which the Quran was responding.
Furthermore, ``legality'' in Islam is not an abstract or primitive
notion relying solely on some alleged ``literal meaning'' of isolated
words, phrases or sentences. The Islamic Legal System explicitly
encompasses: (a) the Prophet's statements (pbuh) and actions related to
the verses, (b) valid interpretations derived from Arabic linguistic
rules, (c) the consensus and opinions of the Prophet's companions and
subsequent scholars, and (d) ongoing rational deduction. In the studied
case: The overwhelming majority of Muslim scholar opinions are in favor
of differentiating between the apparently similar cases of Kalala. By
ignoring this established Islamic legal fact, without giving a rational
argument backing his opinion up, and doing this in favor of a completely
invented scenario---claiming the companions recognized a textual
mistake and conspired to ``fix'' it---Powers constructs a textbook
straw-man argument.

Second, Powers' linguistic claims are factually false. The pretense that
there are no syntactic or linguistic differences between the two verses
is easily refuted by the text itself. As noted, verse 176 begins with
\textarabic{يستفتونك} ``yastaftunaka'', indicating a concrete question
to which the Quran is responding. Furthermore, verse 176 uses the
definite form \textarabic{الكلالة} (\textit{Al-Kalalah}), whereas verse
12 uses the indefinite form \textarabic{كلالة} (\textit{Kalalah}). In
Arabic grammar, this shift from definite to indefinite opens up the
distinct possibility that verse 176 may be handling a general or default
case, while verse 12 is handling a more specific, localized instance.
Additionally, Powers ignores a basic linguistic fact: in Arabic, as in
any human language, the word ``sibling'' \textarabic{أخ أخت إخوة}
(akh/ikhwah) possesses not one, but three possible meanings (full
siblings, uterine siblings, or consanguine siblings). There is neither a
linguistic nor a logical justification for the assumption that verses 12
and 176 use exactly the same meaning, especially considering that verse
176 was revealed before verse 12 and under entirely different
circumstances. Strictly speaking, the three meanings of ``siblings''
yield 9 possible interpretations across both verses combined---not just
the single, unified interpretation Powers forces onto the text.

Yet, and most relevant to this work, even if one, for the sake of
argument, were to accept Powers' syntactic premise, his conclusion rests
on a fundamental misunderstanding of formal axiomatic systems and of the
relationship between an axiomatic system and its semantic
interpretations. Ambiguity in the descriptive language from which an
informal axiomatic system is extracted is not, by itself, a logical
defect. On the contrary, semantic flexibility is a natural feature of
interpretation: the same collection of descriptive statements may admit
multiple semantic models, and the task of formalization is precisely to
identify the structures satisfying the intended constraints.

This observation is particularly important in the present case. The
natural-language expression ``siblings'' \textarabic{إخوة} is not
intrinsically restricted to a single biological relation. In classical
Arabic, as well as many other natural languages, the corresponding terms
can refer to full siblings, uterine siblings, or consanguine siblings.
Consequently, there are at least three linguistically available
interpretations of ``siblings'' in each of the two Kalalah passages.
Since nothing in the wording of verses 12 and 176 by itself forces
exactly one of these three meanings in each occurrence, the combined
text admits nine possible pairings of these interpretations. Powers'
argument effectively selects one of these possible semantic assignments
and then treats the contradiction arising under that assignment as a
contradiction in the underlying system itself. From the perspective of
model theory, this does not follow. What Powers' analysis actually
identifies, at most, is an inconsistency arising under a particular
semantic interpretation of an ambiguous natural-language description. It
does not establish that the underlying formalized system is
inconsistent.

The distinction is fundamental: a contradiction in one interpretation is
not equivalent to unsatisfiability of the axioms. To establish
inconsistency of a formal system, one must show that no semantic model
satisfies its axioms. If another admissible interpretation yields a
coherent model, the contradiction belongs to the rejected interpretation
rather than necessarily to the axioms themselves.

Indeed, ambiguity need not be an imperfection that a formal system must
eliminate. It may reside entirely in the descriptive linguistic layer
while the underlying formal structure remains perfectly coherent once an
interpretation is selected. In the present case, the three ordinary
linguistic meanings of ``siblings'', combined across the two verses,
already show that Powers' choice of a single unified biological meaning
is an additional semantic assumption rather than something forced by the
linguistic data. The resulting contradiction therefore cannot be
presented as a contradiction of the formal inheritance system without
first establishing that this particular interpretation is uniquely
compelled by the linguistic evidence.

This distinction is not foreign to the epistemological framework of the
Quran itself. The Quran explicitly recognizes that its descriptive
language does not consist exclusively of expressions having one
immediately determinate interpretation. In Surah Al-Imran (3:7), it
states:

\begin{quote}
\textarabic{هُوَ الَّذِي أَنزَلَ عَلَيْكَ الْكِتَابَ مِنْهُ آيَاتٌ مُّحْكَمَاتٌ هُنَّ أُمُّ الْكِتَابِ وَأُخَرُ مُتَشَابِهَاتٌ ۖ فَأَمَّا الَّذِينَ فِي قُلُوبِهِمْ زَيْغٌ فَيَتَّبِعُونَ مَا تَشَابَهَ مِنْهُ ابْتِغَاءَ الْفِتْنَةِ وَابْتِغَاءَ تَأْوِيلِهِ ۗ وَمَا يَعْلَمُ تَأْوِيلَهُ إِلَّا اللَّهُ ۗ وَالرَّاسِخُونَ فِي الْعِلْمِ يَقُولُونَ آمَنَّا بِهِ كُلٌّ مِّنْ عِندِ رَبِّنَا ۗ وَمَا يَذَّكَّرُ إِلَّا أُولُو الْأَلْبَابِ}
\end{quote}

A translation is:

\begin{quote}
``It is He who has sent down to you, [O Muhammad], the Book; in it are
verses [that are] precise -- they are the foundation of the Book -- and
others unspecific. As for those in whose hearts is deviation [from
truth], they will follow that of it which is unspecific, seeking discord
and seeking an interpretation [suitable to them]. And no one knows its
[true] interpretation except Allah. But those firm in knowledge say, `We
believe in it. All [of it] is from our Lord.' And no one will be
reminded except those of understanding.''
\end{quote}

The epistemological significance of this statement extends beyond any
particular legal or historical question. The verse explicitly
distinguishes between \textarabic{آيات محكمات} (\textit{ayat
muhkamat}), which are described as \textarabic{أم الكتاب} (the
foundation of the Book), and \textarabic{آيات متشابهات} (\textit{ayat
mutashabihat}), whose interpretation is not presented as being equally
immediate or determinate. The Quran therefore does not establish the
epistemological principle that every textual statement must possess a
single unambiguous surface interpretation. Rather, it explicitly
incorporates different degrees of semantic determinacy into its own
account of knowledge.

More importantly, this verse does not treat the existence of such
passages as evidence that the Book is defective. Instead, it identifies
a distinction between the character of the textual statement and the
epistemic conduct of the interpreter. Those described as
\textarabic{الراسخون في العلم} (\textit{al-rasikhuna fi al-'ilm}, those
firmly grounded in knowledge) are explicitly distinguished from those
who pursue the ambiguous portions seeking discord or an interpretation
of their own choosing. The Quranic epistemological lesson is therefore
that ambiguity in a descriptive layer calls for disciplined
interpretation rather than automatic classification as contradiction.

This principle is directly relevant to the Kalalah problem. The fact
that the expression ``siblings'' admits multiple ordinary linguistic
meanings does not constitute evidence that the inheritance axioms are
contradictory. It establishes that the descriptive language admits
multiple semantic assignments. A rigorous analysis must therefore
distinguish the ambiguity of the linguistic representation from the
consistency of the formal structure represented by it. The Quranic
distinction between \textit{muhkamat} and \textit{mutashabihat}
provides an internal epistemological basis for making precisely this
distinction. The Kalalah passages provide a concrete example of this
broader principle. The natural-language description permits multiple
semantic models, while the Islamic jurisprudential tradition supplies
interpretations under which the resulting inheritance rules form
coherent systems.

Consequently, the mere existence of an interpretation under which the
shares appear contradictory does not establish that the Quranic
inheritance axioms are contradictory. It establishes, at most, that one
selected interpretation fails to produce a coherent model.

This distinction also exposes the epistemological boundary in Powers'
argument. By treating an ambiguity of descriptive natural language as
though it were already a contradiction in the formal system described by
that language, the analysis moves from a linguistic observation to a
mathematical conclusion without establishing the necessary semantic
bridge. Model-theoretically, the correct procedure is the reverse:
identify the possible interpretations, formalize the corresponding
semantic structures, and determine whether the axioms are satisfiable
under an admissible interpretation. Under this criterion, ambiguity is a
feature of the descriptive layer, not evidence by itself of
inconsistency in the formal layer.

The point is therefore not that ambiguity automatically validates every
possible interpretation. Rather, it is that ambiguity must first be
recognized as ambiguity. A particular interpretation may subsequently be
rejected on linguistic, legal, mathematical, or empirical grounds; but
the existence of a contradiction within that interpretation cannot
itself establish a contradiction in the underlying system unless the
interpretation is shown to be uniquely compelled by the available
evidence, which Powers cannot show.

Similarly, confronted with dynamic \textit{Awl} scenarios, which shall be
discussed in detail in the next sections, where overlapping shares yield
cumulative fractions exceeding a value of one, some methodologically
restricted critics routinely declare the existence of an elementary
arithmetic error \cite{Burton1990, Crone1987}. In doing so, they mistake
a highly sophisticated system of relational, scaling variables for a
primitive, poorly calculated physical pie. By treating their
non-scientific, qualitative tools as an all-encompassing framework, such
participants fail to see where their model's scope ends and where a
precise mathematical phenomenon begins.

The following sections audit the text's exact syntactic
architecture---specifically the definitive role of
\textit{al-ta'r\={i}f} and the partitive alternation between ``ma''
(\textarabic{ما}) and ``min ma'' (\textarabic{من ما})---to demonstrate
that the Quranic axioms explicitly bypass basic arithmetic in favor of an
advanced algebraic engine. By embedding proportional relational weights
alongside absolute biological fail-safes that eliminate mathematical and
logical contradiction, the text operates flawlessly beyond the reach of
historicist reduction, reversing the burden of explanation onto
participants of the ``Islamic studies'' discipline that has mistaken its
narrow interpretive scope for a universal explanatory tool.

Without excluding the possibility of external historical influences, the
formal analysis presented here demonstrates also that
Al-Khw\={a}rizm\={i} had the essential raw conceptual ingredients of his
Algebra laid out directly before him within the linguistic architecture
of these three Quranic verses \cite{Oaks2014, Rashed1994}.

\subsection{The respective Quranic verses and their interpretation}

The foundation is encapsulated in three pivotal verses of Sura 4,
Al-nisa. A translation (according to: Yousuf Ali) of the relevant
portions is as follows:

\begin{itemize}
    \item \textbf{Verse 11:} \\
    \textarabic{(يُوصِيكُمُ اللَّهُ فِي أَوْلَادِكُمْ لِلذَّكَرِ مِثْلُ حَظِّ الْأُنْثَيَيْنِ فَإِنْ كُنَّ نِسَاءً فَوْقَ اثْنَتَيْنِ فَلَهُنَّ ثُلُثَا مَا تَرَكَ وَإِنْ كَانَتْ وَاحِدَةً فَلَهَا النِّصْفُ وَلِأَبَوَيْهِ لِكُلِّ وَاحِدٍ مِنْهُمَا السُّدُسُ مِمَّا تَرَكَ إِنْ كَانَ لَهُ وَلَدٌ فَإِنْ لَمْ يَكُنْ لَهُ وَلَدٌ وَوَرِثَهُ أَبَوَاهُ فَلِأُمِّهِ الثُّلُثُ فَإِنْ كَانَ لَهُ إِخْوَةٌ فَلِأُمِّهِ السُّدُسُ مِنْ بَعْدِ وَصِيَّةٍ يُوصِي بِهَا أَوْ دَيْنٍ آَبَاؤُكُمْ وَأبناؤكم لَا تَدْرُونَ أَيُّهُمْ أَقْرَبُ لَكُمْ نَفْعًا فَرِيضَةً مِنَ اللَّهِ إِنَّ اللَّهَ كَانَ عَلِيمًا حَكِيمًا)} \\
    ``Allah (thus) directs you as regards your Children's (Inheritance):
    to the male, a portion equal to that of two females: if only
    daughters, two or more, their share is two-thirds of the
    inheritance; if only one, her share is a half. For parents, a sixth
    share of the inheritance to each, if the deceased left children; if
    no children, and the parents are the (only) heirs, the mother has a
    third; if the deceased Left brothers (or sisters) the mother has a
    sixth. (The distribution in all cases ('s) after the payment of
    legacies and debts. Ye know not whether your parents or your
    children are nearest to you in benefit. These are settled portions
    ordained by Allah; and Allah is All-knowing, Al-wise.''

    \item \textbf{Verse 12:} \\
    \textarabic{(وَلَكُمْ نِصْفُ مَا تَرَكَ أزواجكُمْ إِنْ لَمْ يَكُنْ لَهُنَّ وَلَدٌ فَإِنْ كَانَ لَهُنَّ وَلَدٌ فَلَكُمُ الرُّبُعُ مِمَّا تَرَكْنَ مِنْ بَعْدِ وَصِيَّةٍ يُوصِينَ بِهَا أَوْ دَيْنٍ وَلَهُنَّ الرُّبُعُ مِمَّا تَرَكْتُمْ إِنْ لَمْ يَكُنْ لَكُمْ وَلَدٌ فَإِنْ كَانَ لَكُمْ وَلَدٌ فَلَهُنَّ الثُّمُنُ مِمَّا تَرَكْتُمْ مِنْ بَعْدِ وَصِيَّةٍ تُوصُونَ بِهَا أَوْ دَيْنٍ وَإِنْ كَانَ رَجُلٌ يُورَثُ كَلَالَةً أَوِ امْرَأَةٌ وَلَهُ أَخٌ أَوْ أُخْتٌ فَلِكُلِّ وَاحِدٍ مِنْهُمَا السُّدُسُ فَإِنْ كَانُوا أَكْثَرَ مِنْ ذَلِكَ فَهُمْ شُرَكَاءُ فِي الثُّلُثِ مِنْ بَعْدِ وَصِيَّةٍ يُوصَى بِهَا أَوْ دَيْنٍ غَيْرَ مُضَارٍّ وَصِيَّةً مِنَ اللَّهِ وَاللَّهُ عَلِيمٌ حَلِيمٌ)} \\
    ``In what your wives leave, your share is a half, if they leave no
    child; but if they leave a child, ye get a fourth; after payment of
    legacies and debts. In what ye leave, their share is a fourth, if ye
    leave no child; but if ye leave a child, they get an eighth; after
    payment of legacies and debts. If the man or woman whose inheritance
    is in question, has left neither ascendants nor descendants, but has
    left a brother or a sister, each one of the two gets a sixth; but if
    more than two, they share in a third; after payment of legacies and
    debts; so that no loss is caused (to any one). Thus is it ordained
    by Allah; and Allah is All-knowing, Most Forbearing.''

    \item \textbf{Verse 176:} \\
    \textarabic{(يَسْتَفْتُونَكَ قُلِ اللّهُ يُفْتِيكُمْ فِي الْكَلاَلَةِ إِنِ امْرُؤٌ هَلَكَ لَيْسَ لَهُ وَلَدٌ وَلَهُ أُخْتٌ فَلَهَا نِصْفُ مَا تَرَكَ وَهُوَ يَرِثُهَا إِن لَّمْ يَكُن لَّهَا وَلَدٌ فَإِن كَانَتَا اثْنَتَيْنِ فَلَهُمَا الثُّلُثَانِ مِمَّا تَرَكَ وَإِن كَانُواْ إِخْوَةً رِّجَالاً وَنِسَاء فَلِلذَّكَرِ مِثْلُ حَظِّ الأُنثَيَيْنِ يُبَيِّنُ اللّهُ لَكُمْ أَن تَضِلُّواْ وَاللّهُ بِكُلِّ شَيْءٍ عَلِيمٌ)} \\
    ``They ask thee for a legal decision. Say: Allah directs (thus)
    about those who leave no descendants or ascendants as heirs. If it
    is a man that dies, leaving a sister but no child, she shall have
    half the inheritance: If (such a deceased was) a woman, who left no
    child, Her brother takes her inheritance: If there are two sisters,
    they shall have two-thirds of the inheritance (between them): if
    there are brothers and sisters, (they share), the male having twice
    the share of the female. Thus doth Allah make clear to you (His
    law), lest ye err. And Allah hath knowledge of all things.''
\end{itemize}

A rigorous linguistic analysis of these verses reveals a crucial
structural characteristic regarding how proportions are defined. Across
these verses, there are precisely fifteen instances where inheritance
fractions are apportioned, respectively:

\begin{enumerate}
    \item \textarabic{لِلذَّكَرِ مِثْلُ حَظِّ الْأُنْثَيَيْنِ}
    \item \textarabic{فَلَهُنَّ ثُلُثَا مَا تَرَكَ}
    \item \textarabic{فَلَهَا النِّصْفُ}
    \item \textarabic{لِكُلِّ وَاحِدٍ مِنْهُمَا السُّدُسُ مِمَّا تَرَكَ}
    \item \textarabic{فَلِأُمِّهِ الثُّلُثُ}
    \item \textarabic{فَلِأُمِّهِ السُّدُسُ}
    \item \textarabic{نِصْفُ مَا تَرَكَ أزواجكُمْ}
    \item \textarabic{فَلَكُمُ الرُّبُعُ مِمَّا تَرَكْنَ}
    \item \textarabic{وَلَهُنَّ الرُّبُعُ مِمَّا تَرَكْتُمْ}
    \item \textarabic{فَلَهُنَّ الثُّمُنُ مِمَّا تَرَكْتُمْ}
    \item \textarabic{فَلِكُلِّ وَاحِدٍ مِنْهُمَا السُّدُسُ}
    \item \textarabic{فَهُمْ شُرَكَاءُ فِي الثُّلُثِ}
    \item \textarabic{فَلَهَا نِصْفُ مَا تَرَكَ}
    \item \textarabic{فَلَهُمَا الثُّلُثَانِ مِمَّا تَرَكَ}
    \item \textarabic{فَلِلذَّكَرِ مِثْلُ حَظِّ الأُنثَيَيْنِ}
\end{enumerate}

In ten of these fifteen instances, the Quran defines the proportions
directly using the definite article \textarabic{ال} (\textit{al-ta'r\={i}f}),
utilizing formulations such as ``for her is the half''
(\textarabic{لَهَا النِّصْفُ}), ``for you is the fourth''
(\textarabic{وَلَكُمُ الرُّبُعُ}), or ``for his mother is the third''
(\textarabic{فَلِأُمِّهِ الثُّلُثُ}). In only three out of the fifteen
instances does the Quran link the proportion to the total value of the
estate using an additive genitive construction
(\textarabic{الإضافة إلى القيمة الكلية}), such as ``two thirds of what
he left'' (\textarabic{فَوْقَ اثْنَتَيْنِ فَلَهُنَّ ثُلُثَا مَا تَرَكَ}).
The two remaining instances are actually a repetition of the
mathematical assertion applied to offsprings and siblings only: ``to the
male, a portion equal to that of two females''.

This grammatical predominance is highly significant. In Arabic, the use
of the definite article presents the prescribed share as a determinate
and independently identifiable conceptual entity in the discourse. Thus,
in expressions such as \textarabic{لَهَا النِّصْفُ} (``for her is the
half''), \textarabic{فَلِأُمِّهِ الثُّلُثُ} (``for his mother is the
third''), and \textarabic{فَلَكُمُ الرُّبُعُ} (``for you is the
fourth''), the fraction is grammatically presented as an already
identifiable share, without its denotation being syntactically
established by reference to the actual numerical value of the estate.
The share therefore exists at the conceptual level independently of any
particular estate value to which it may subsequently be applied.

This contrasts with constructions in which the fraction is explicitly
linked to the estate through a genitive or relative construction, such
as \textarabic{ثُلُثَا مَا تَرَكَ} (``two thirds of what he left''). In
such cases, \textarabic{ما تَرَكَ} explicitly supplies the reference
domain against which the fraction is evaluated. The distinction is
therefore not merely stylistic: the former construction introduces the
prescribed fraction as a determinate conceptual quantity, while the
latter explicitly relates that quantity to the concrete estate described
by the text.

Accordingly, the predominance of definite, unlinked fractions provides
linguistic evidence for distinguishing the conceptual share from the
numerical value of the estate itself. The prescribed fractions can first
be represented as abstract relational quantities and only subsequently
instantiated within a particular inheritance configuration. This
distinction becomes especially important in configurations involving
overlapping prescribed shares, where the conceptual shares may
collectively exceed unity and therefore require an algebraic
determination of their realization over the actual estate.

\subsection{The meaning of ``ma'' (\textarabic{ما}) and ``min ma'' (\textarabic{من ما})}

Building upon this linguistic framework, a further syntactic distinction
must be made regarding the way the Quran mentions the relation of the
proportions to the estate itself. The Quran systematically alternates
between two distinct linguistic constructions that bear different
linguistic meanings.

The first utilizes the relative pronoun: \textarabic{ما} ``ma''
(what/that which). Linguistically, ``ma'' denotes obscurity and
totality. Mathematically, it binds the prescribed fraction to the
absolute total of the estate.

If we say, for example: \textarabic{``أعطني نصف ما تملك''} (English:
``give me half of what you possess'') we do mean, as in English also,
exactly half of the total amount of the possession.

The second construction introduces the preposition of partitioning:
\textarabic{من} ``min,'' forming: \textarabic{من ما} ``min ma'' (from
what/out of what). In Arabic syntax, the preposition ``min'' here is
partitive (\textit{tab'idiyyah}), linguistically permitting the fraction
to be calculated from a subset or a partitioned portion of the whole,
rather than the absolute total.

When we say: \textarabic{``أعطني النصف من ما تملك''} (English: ``give me
the half from what you possess'') we mean, also in English, that you
need to give me ``the agreed upon half'' (marked by the definite
article) from what you possess, which does not mean, as in the previous
sentence, half of all your possessions.

This subtle grammatical distinction provides the foundational mechanism
for consistently resolving complex inheritance scenarios where multiple
shares overlap.

\subsection{The deduced notion of a math variable}
\label{subsec:variable}

The semantic framework established in the first two sections, when
formalized, inevitably gives rise to the notion of a ``mathematical
variable'', elegantly captured by the pronoun \textarabic{ما} ``ma'',
which linguistically also denotes: obscurity as well as totality.

An interesting fact worth mentioning in this context is, that
Al-Khw\={a}rizm\={i} in his famous Algebra book did not find himself in
need of translating the Quranic assertions to any symbolic Language
utilizing math variables, and just referred to the \textarabic{ما}
``ma''-constructions utilizing the Arabic word: \textarabic{``مال''}
(\textit{belongings}), i.e., he used an Arabic word to stand for what we
now consider to be a math variable, without any loss of generality.

The importance of the use of math variables (or their semantic
equivalents in Arabic) to express Quranic inheritance rules is not seen
in trivial cases where all prescribed proportions add up to the full
percentage of the estate, but when $1/2$, $1/3$, and/or $2/3$ fractions
seem to cause a math contradiction.

For example, if a deceased woman leaves behind a husband and two
sisters, the prescribed shares are $1/2$ and $2/3$ respectively, summing
to $7/6$ (116.6\%).

Historically, Muslim jurists after Al-Khw\={a}rizm\={i} resolved this by
treating the conceptual total percentage as an unknown variable ($X$)
that must scale down proportionally. The equation becomes:
$1/2(X)+2/3(X)=1$, necessitating the use of linear equations.

This amounted to treating \textarabic{ما} ``ma'' and \textarabic{من ما}
``min-ma'' constructions as denoting the same, namely: ``the whole or
less''.

Alternatively, if one applies a strict literal reading (like the one
reported to be the opinion of Ibn-Abbas), distinguishing between
\textarabic{ما} ``ma'' and \textarabic{من ما} ``min ma'', the husband
takes his absolute $1/2$, and the sisters' $2/3$ share is calculated
from a theoretical subset ($p$), requiring an equation like
$2/3(p)=\text{remaining estate}$.

In both methodologies, the text forces the abandonment of basic
arithmetic in favor of algebraic variables to ensure a flawless 100\%
distribution.

\subsection{Consistence of the inheritance axioms}

A potential critique of the literal reading is that a mathematical
contradiction would occur if two shares---such as a $1/2$ share and a
$2/3$ share---both utilizing \textarabic{ما} ``ma'' were to co-occur.

However, the Quranic axioms exhibit a structural consistency safeguarded
by biological and logical mutual exclusivity.

Instances where heirs receive $1/2$ using \textarabic{ما} ``ma'' include
a husband (when there are no children), a sister (when there are no
children), and a single daughter. Conversely, the $2/3$ share using
\textarabic{ما} ``ma'' is prescribed exclusively for two or more
daughters.

A husband only receives $1/2$ in the absence of children, rendering it
biologically impossible for his share to coexist with the $2/3$ share of
two daughters.

Similarly, a sister's $1/2$ share precludes the existence of children,
and the existence of a single daughter logically precludes the existence
of ``two or more'' daughters simultaneously.

Therefore, the axioms of the system are structurally immune to producing
a literal mathematical contradiction, as the text embeds biological
fail-safes to prevent logical errors.

\subsection{A challenge for de-mystification efforts}

Although the aim of the current work is to provide a logical basis for
Al-Khw\={a}rizm\={i}'s Algebra, it is worth mentioning that the
structural and mathematical realities of the Inheritance verses
presented here provide a definitive challenge for attempts of secular
de-mystification of the above Qur'anic verses. Some of those attempts
use the literal reading to claim the existence of a contradiction in the
Quranic Inheritance axioms.

As we saw in the last sections: Whether one applies a unified
jurisprudential reading that necessitates an unknown variable ($X$) for
proportional reduction, or a strict literal reading that differentiates
between \textarabic{ما} ``ma'' and \textarabic{من ما} ``min ma'' to
necessitate a variable ($p$) for a theoretical subset, the text operates
flawlessly. Furthermore, the text pre-emptively safeguards itself
against literal mathematical contradictions through biological
exclusivity.

This moves the burden of explanation of how such a rigorous and
unprecedented legal framework is a possible achievement of someone
living in the seventh century to the side of de-mystification efforts,
knowing that the framework functions perfectly under two divergent
linguistic interpretations---both secretly requiring a mathematical
science (Algebra) that would not be formalized for another two
centuries---while simultaneously embedding logical fail-safes.

The algebraic equations, dual linguistic readings, and proportional
design embedded within these verses demonstrate a level of systemic
vision that conclusively demonstrates unexplained capacities.

\subsection{Conclusion}

The preceding analysis reveals that the Quran employs a profoundly
efficient, informal diction to describe the fundamental axioms of
inheritance. This linguistic architecture is characterized by the
following structural features:

\begin{itemize}
    \item \textbf{Promoting the notion of a mathematical Variable:} It
    encourages the use of mathematical variables bound to structural
    constraints by utilizing linguistic constructs that semantically
    denote totality, obscurity, and partial inclusion---specifically
    through the alternation of \textarabic{ما} ``ma'' and
    \textarabic{من ما} ``min ma''.

    \item \textbf{Defining proportions relative to the total amount of
    shares:} The predominance of the definite article
    (\textit{al-ta'r\={i}f}) establishes the shares as proportional
    percentages of a conceptual whole, rather than fixed arithmetic
    extractions from a physical estate. Direct sentences relating
    proportions to each other, like \textarabic{للذكر مثل حظ الأنثيين}
    (to a male a portion equal to that of two females) provides more
    evidence for that.

    \item \textbf{Semantic Flexibility and Historical Evolution:} It
    opens the door to multiple valid linguistic interpretations (such as
    the biological categorization of ``siblings'' or the strict versus
    unified treatment of grammatical particles), each yielding a
    distinct inheritance-allocation mechanism. Historically, this
    generated three major developmental phases: the strict literal
    system of Ibn Abbas, the unified jurisprudential system of Umar, and
    ultimately, the formalized algebraic system of
    Al-Khw\={a}rizm\={i}.

    \item \textbf{Internal Consistency:} Each of these distinct
    methodologies is internally consistent and can be soundly linked to
    both legal and rational foundations.

    \item \textbf{Logical Immunity:} It achieves this remarkable
    interpretive flexibility while remaining immune to logical
    inconsistency, safeguarded by embedded biological fail-safes.

    \item \textbf{Compactness:} It accomplishes all of this in the most
    compact way possible, encoding a scalable algebraic framework capable
    of representing arbitrarily large finite inheritance configurations
    within merely three verses.
\end{itemize}

In light of this structural sophistication, the prevalent secular and
historicist critiques of the Quranic inheritance system---exemplified by
``Islamic studies'' participants such as David Powers, John Burton, and
Patricia Crone---suffer from a critical epistemological boundary
failure. By attempting to evaluate a sophisticated mathematical
framework using purely qualitative and philological tools, these critics
mistake their own restrictive interpretive models for inherent flaws
within the text. Whether addressing the alleged legal incoherence of the
\textit{Kalalah} scenarios or the supposed arithmetic errors of
\textit{Awl} (oversubscription), their claims collapse under formal
logical scrutiny. As known from formal model theory, demonstrating that
a system yields a contradiction in a single, selected semantic model
does not render the underlying axioms unsatisfiable; it merely
necessitates the discarding of the flawed model.

Whether one applies the unified jurisprudential reading---which requires
a scaling variable ($X$) to resolve the \textit{Awl} overflow---or the
strict literal reading---which requires a theoretical subset variable
($p$)---the text forces the abandonment of basic arithmetic in favor of
linear algebraic equations.

Ultimately, the burden of explanation of this phenomenon is decisively
reversed to become the task of western, secular de-mystification efforts
in ``Islamic studies'': The raw conceptual ingredients of Algebra were
not retroactively imposed upon the text; rather, they were directly laid
out within its linguistic architecture. It was this very architecture
that may have compelled Al-Khw\={a}rizm\={i} to formalize the science of
Algebra, cementing the Quranic inheritance system as a mathematically
robust, structurally flawless framework.

These linguistic features motivate, but do not by themselves entail,
the formal construction that follows. Sections~\ref{sec:rep}--\ref{sec:dec}
supply the formal content: a representation lemma, a canonicity theorem,
and a decision procedure for the inheritance-relevant fragment.


\section{Al-Khw\={a}rizm\={i}'s Inheritance Algebra}
\label{sec:rep}

\subsection{Inheritance configurations and the representation lemma}

\begin{definition}[Inheritance configuration]
An \emph{inheritance configuration} is a tuple $(H,R,W,G)$ where $H$ is
a finite set of heirs drawn from a fixed finite set of heir classes,
$R$ is a finite set of applicability conditions (each a Boolean
combination of biological preconditions such as presence or absence of
children, parents, and siblings of each class), $W$ is a finite set of
\textit{wa\d{s}iyya} (bequest) conditions, and $G\in\mathbb{Q}_{>0}$ is
the gross estate.
\end{definition}

\begin{definition}[Admissible allocation]
An \emph{admissible allocation} for $(H,R,W,G)$ is a map
$a\colon H\to\mathbb{Q}_{\geq 0}$ together with a bequest value
$w\in\mathbb{Q}_{\geq 0}$ such that
(i) the prescribed shares match the Qur'anic fractions under the
applicable $R$-conditions,
(ii) $\sum_{h\in H} a(h)+w=G$ (conservation), and
(iii) $w\leq G/3$ (legal cap).
\end{definition}

\begin{lemma}[Representation]
\label{lem:rep}
Every inheritance configuration is representable as a finite system
\[
\mathcal{S}=\{\,a_j^{\mathsf T}x\leq b_j : 1\leq j\leq m\,\},
\qquad
a_j\in\mathbb{Q}^n,\quad b_j\in\mathbb{Q},
\]
whose rational solutions correspond bijectively to admissible
allocations.
\end{lemma}

\begin{proof}[Proof sketch]
Normalize $G=1$: since $G\in\mathbb{Q}$, scaling by $G^{-1}$ preserves
rationality and the solution set is the same up to the same scaling.
Each prescribed share is a rational constant
($1/2$, $1/3$, $2/3$, $1/4$, $1/6$, $1/8$) applied to the estate or to
a partitive subset; each such prescription is a linear equality in the
unknowns. Each applicability condition is a finite case distinction,
generating one linear system per case; the union over cases is the
finite family of systems whose disjunction represents the configuration.
Conservation is the equality $\sum_i x_i+w=1$. The \textit{wa\d{s}iyya}
cap is the inequality $w\leq 1/3$. Equalities are replaced by pairs of
inequalities.
\end{proof}

\begin{remark}[Rationality]
Because shares, caps, and normalization are rational, all coefficients
$a_j,b_j$ are rational. This is what licenses Fourier--Motzkin
elimination over $\mathbb{Q}$ rather than over $\mathbb{R}$ in
\S\ref{sec:dec}.
\end{remark}

\begin{remark}[Finitely many cases]
The case space is finite because heir classes and preconditions are
drawn from a fixed finite vocabulary; this is the formal counterpart of
\textit{al-\d{s}abr wa-al-taqs\={i}m} invoked in
\S\ref{sec:fiqh}.
\end{remark}

\subsection{A worked bequest}
\label{sec:worked}

Consider a classical inheritance problem in which a man dies leaving
two sons and wills to a stranger a legacy equal to one third of the
estate minus one fourth of what one son receives. Normalize the gross
estate to
\[
G=1.
\]
Let
\[
x\in\mathbb{Q}
\]
denote the share of one son and
\[
w\in\mathbb{Q}
\]
denote the value of the bequest. The corresponding system is
\begin{align}
2x+w &= 1,
&&\text{(Estate Conservation)},\\
w &= \frac{1}{3}-\frac{1}{4}x,
&&\text{(Bequest Definition)},\\
w &\leq \frac{1}{3},
&&\text{(Legal Upper Bound)}.
\end{align}

The equality constraints may be represented by pairs of inequalities:
\begin{equation}
\begin{aligned}
w &\leq 1-2x,
&
w &\geq 1-2x,\\
w &\leq \frac{1}{3}-\frac14x,
&
w &\geq \frac{1}{3}-\frac14x.
\end{aligned}
\end{equation}

Fourier--Motzkin elimination combines every lower bound of $w$ with
every upper bound of $w$. The two nontrivial projected constraints are
\begin{align}
1-2x &\leq \tfrac13-\tfrac14 x
&&\Longrightarrow&
x &\geq \tfrac{8}{21},
\\
\tfrac13-\tfrac14 x &\leq 1-2x
&&\Longrightarrow&
x &\leq \tfrac{8}{21}.
\end{align}
Together they give $x=\tfrac{8}{21}$ directly, without appealing to the
retained equality. The legal cap $w\leq\tfrac13$ has already been used
in forming the first projection.

The projected system therefore determines $x$ uniquely, and
substitution recovers
\begin{equation}
w
=
\frac13-\frac14\left(\frac{8}{21}\right)
=
\frac5{21}.
\end{equation}
The legal upper bound is satisfied:
\begin{equation}
\frac5{21}\leq\frac13.
\end{equation}

The equality information retained from the original system is
consistent with this value; it is not needed to establish uniqueness
once both projections are computed.

Thus the inheritance problem admits the formal transformation
\begin{equation}
\boxed{
\begin{aligned}
\text{jurisprudential problem}
&\longrightarrow
\text{rational linear constraints}
\\
&\longrightarrow
\text{exact computational solution}.
\end{aligned}}
\end{equation}

The use of Fourier--Motzkin elimination here is a modern formal
representation of the computational structure of the problem; it is
not a claim that Al-Khw\={a}rizm\={i} historically formulated his
method using this modern elimination procedure.


\section{The Fixed-Dimensional Inheritance Structure}

For a fixed inheritance configuration with $d$ claimant coordinates,
consider the one-sorted first-order language
\begin{equation}
\mathcal{L}_d
=
\left\langle
0,+,\leq,\{f_q\}_{q\in\mathbb{Q}},
e_1,\ldots,e_d
\right\rangle.
\end{equation}

Here $0$ is the zero allocation, $+$ is vector addition, $\leq$ is the
coordinatewise order, $f_q$ is multiplication by the rational scalar
$q$, and $e_1,\ldots,e_d$ are distinguished basis vectors corresponding
to the claimant coordinates.

The intended structure is
\begin{equation}
\mathcal{V}_d
=
\left(
\mathbb{Q}^d,
0,+,\leq_{\mathrm{cw}},
\{f_q\}_{q\in\mathbb{Q}},
e_1,\ldots,e_d
\right),
\end{equation}
where
\[
e_i=(0,\ldots,0,1,0,\ldots,0)
\]
and
\begin{equation}
(x_1,\ldots,x_d)
\leq_{\mathrm{cw}}
(y_1,\ldots,y_d)
\iff
\bigwedge_{i=1}^{d}x_i\leq y_i.
\end{equation}

\begin{definition}[Admissible structure]
\label{def:adm}
A structure
$\mathcal{M}=(M,0,+,\leq,\{f_q\}_{q\in\mathbb{Q}},e_1,\ldots,e_d)$
in $\mathcal{L}_d$ is \emph{admissible} if:
\begin{enumerate}
\item[(A1)] $M$ is a $\mathbb{Q}$-vector space under $+$ and
$\{f_q\}_{q\in\mathbb{Q}}$;
\item[(A2)] $e_1,\ldots,e_d$ are linearly independent over
$\mathbb{Q}$;
\item[(A3)] every $x\in M$ is a $\mathbb{Q}$-linear combination of
$e_1,\ldots,e_d$ (spanning);
\item[(A4)] for all $x,y\in M$ with $x=\sum_i q_i e_i$ and
$y=\sum_i r_i e_i$, we have $x\leq y$ if and only if $q_i\leq r_i$ for
every $i$ (order is the coordinatewise order induced by the
distinguished basis).
\end{enumerate}
\end{definition}

The class of admissible structures is exactly the class considered in
the canonicity theorem below.


\section{Uniqueness of the Fixed-Dimensional Structure}
\label{sec:canon}

\begin{theorem}[Canonicity]
\label{thm:canon}
Let $\mathcal{M}$ be an admissible structure in the language
$\mathcal{L}_d$ in the sense of Definition~\ref{def:adm}. Then
\[
\mathcal{M}\cong\mathcal{V}_d.
\]
\end{theorem}

\begin{proof}
Let $M$ be the domain of $\mathcal{M}$. By the spanning condition, for
every $x\in M$ there exist rational numbers
$q_1,\ldots,q_d$ such that
\[
x=q_1e_1+\cdots+q_de_d.
\]
By linear independence of
$e_1,\ldots,e_d$, this representation is unique.

Define
\[
\phi:M\longrightarrow\mathbb{Q}^d
\]
by
\[
\phi(x)=(q_1,\ldots,q_d),
\]
where $(q_1,\ldots,q_d)$ is the unique coordinate tuple of $x$.

The map $\phi$ is surjective because every
$(q_1,\ldots,q_d)\in\mathbb{Q}^d$ is represented by
\[
q_1e_1+\cdots+q_de_d.
\]
It is injective because two elements with the same coordinate tuple
have the same representation in the distinguished basis.

If
\[
x=\sum_{i=1}^{d}q_ie_i,
\qquad
y=\sum_{i=1}^{d}r_ie_i,
\]
then
\[
x+y=\sum_{i=1}^{d}(q_i+r_i)e_i,
\]
and therefore
\[
\phi(x+y)
=
\phi(x)+\phi(y).
\]
Likewise, for every $q\in\mathbb{Q}$,
\[
\phi(f_q(x))=q\phi(x).
\]

By the coordinatewise-order axiom,
\[
x\leq y
\iff
q_i\leq r_i
\quad\text{for every }i,
\]
which is equivalent to
\[
\phi(x)\leq_{\mathrm{cw}}\phi(y).
\]
Finally,
\[
\phi(e_i)
\]
is the $i$th standard basis vector of $\mathbb{Q}^d$.

Thus $\phi$ is an isomorphism preserving zero, addition, rational
scalar multiplication, the distinguished basis vectors, and order.
Therefore
\[
\mathcal{M}\cong\mathcal{V}_d.
\]
\end{proof}

\begin{corollary}[Uniqueness within the admissible class]
\label{cor:unique}
Within the class of admissible structures, $\mathcal{V}_d$ is the unique
model up to isomorphism. Consequently, any two admissible rational
linear-coordinate formalizations of the same inheritance configuration
are isomorphic, and their allocations coincide under the isomorphism.
\end{corollary}

\begin{remark}[What canonicity does and does not say]
Canonicity is a property \emph{within} the admissible class. It does not
assert that non-linear or non-coordinatewise formalizations are
impossible, nor that linear algebra is the only methodologically
legitimate tool. It asserts that once the inheritance structure is
represented by rational linear coordinates, that representation is
unique up to isomorphism. This is the precise sense in which the
algebraic representation is canonical, and it is the sense required by
the argument of \S2 against treating a single semantic model as
definitive.
\end{remark}

\begin{remark}[Relation to \S2.3]
The observation in \S2.3 that \textit{`awl} resolves $\sum s_i X=1$ is
consistent with this theorem: \textit{`awl} is a computation in the
same rational linear-coordinate class, expressed in different notation.
Canonicity therefore explains why the pre-algebraic and algebraic
procedures agree on allocations---they are two presentations of the same
admissible structure.
\end{remark}


\section{The First-Order Theory of the Fixed-Dimensional Model}

The structure $\mathcal{V}_d$ defined in the preceding section is a
first-order structure for the language
\[
\mathcal{L}_d
=
\left\langle
0,+,\leq,\{f_q\}_{q\in\mathbb{Q}},
e_1,\ldots,e_d
\right\rangle.
\]
The structure itself is not a first-order sentence. To associate a
first-order theory with $\mathcal{V}_d$, we consider all first-order
sentences that can be formed in $\mathcal{L}_d$ and retain precisely
those that are true in $\mathcal{V}_d$.

Let
\[
\operatorname{Sent}(\mathcal{L}_d)
\]
denote the set of all first-order sentences in the language
$\mathcal{L}_d$. The complete first-order theory of the intended
inheritance structure is defined by
\begin{equation}
T_d
=
\operatorname{Th}_{\mathcal{L}_d}(\mathcal{V}_d)
=
\left\{
\varphi\in\operatorname{Sent}(\mathcal{L}_d)
:
\mathcal{V}_d\models\varphi
\right\}.
\end{equation}

Thus the passage from the intended structure to its first-order theory
is the semantic construction
\begin{equation}
\boxed{
\mathcal{V}_d
\quad\longrightarrow\quad
\operatorname{Sent}(\mathcal{L}_d)
\quad\longrightarrow\quad
T_d
=
\left\{
\varphi:
\mathcal{V}_d\models\varphi
\right\}.
}
\end{equation}

Equivalently, for every
$\varphi\in\operatorname{Sent}(\mathcal{L}_d)$,
\begin{equation}
\mathcal{V}_d\models\varphi
\quad\Longleftrightarrow\quad
\varphi\in T_d.
\end{equation}

Hence $T_d$ contains exactly the first-order statements in
$\mathcal{L}_d$ that hold in the intended fixed-dimensional
inheritance structure.

\subsubsection{An Explicit First-Order Axiom System}

Although $T_d$ is defined semantically, the principal algebraic
properties of the intended structure can be stated explicitly as
first-order axioms in the same one-sorted language.

Let $A_d$ denote the first-order axiom system consisting of the
following sentences.

First, the additive structure is an abelian group:
\begin{align}
\forall x\,\forall y\,\forall z\;&
((x+y)+z=x+(y+z)),
\tag{$A_1$}\\
\forall x\,\forall y\;&
(x+y=y+x),
\tag{$A_2$}\\
\forall x\;&
(x+0=x),
\tag{$A_3$}\\
\forall x\;&
\exists y\;(x+y=0).
\tag{$A_4$}
\end{align}

For every fixed $q,r\in\mathbb{Q}$, the scalar-operation symbols
$f_q$ satisfy
\begin{align}
\forall x\;& f_1(x)=x,
\tag{$A_5$}\\
\forall x\;& f_0(x)=0,
\tag{$A_6$}\\
\forall x\;& f_q(f_r(x))=f_{qr}(x),
\tag{$A_7$}\\
\forall x\;& f_{q+r}(x)=f_q(x)+f_r(x),
\tag{$A_8$}\\
\forall x\,\forall y\;&
f_q(x+y)=f_q(x)+f_q(y).
\tag{$A_9$}
\end{align}

Here $q$ and $r$ are not variables of the object language. They range
over the externally fixed coefficient field $\mathbb{Q}$, while each
corresponding operation symbol $f_q$ belongs to the language.

The order relation satisfies the first-order axioms of a partial order:
\begin{align}
\forall x\;& x\leq x,
\tag{$A_{10}$}\\
\forall x\,\forall y\;&
\bigl((x\leq y)\land(y\leq x)\bigr)\rightarrow x=y,
\tag{$A_{11}$}\\
\forall x\,\forall y\,\forall z\;&
\bigl((x\leq y)\land(y\leq z)\bigr)\rightarrow x\leq z.
\tag{$A_{12}$}
\end{align}

The order is compatible with addition:
\begin{equation}
\forall x\,\forall y\,\forall z\;
\bigl(x\leq y\rightarrow x+z\leq y+z\bigr).
\tag{$A_{13}$}
\end{equation}

For every positive rational $q>0$, scalar multiplication preserves
the order:
\begin{equation}
\forall x\,\forall y\;
\bigl(x\leq y\rightarrow f_q(x)\leq f_q(y)\bigr).
\tag{$A_{14}$}
\end{equation}

The distinguished vectors are required to be linearly independent
over the fixed rational coefficient field. Since every rational scalar
has a corresponding function symbol, this condition is represented by
the family of first-order sentences
For every nonzero tuple
\[
(q_1,\ldots,q_d)\in\mathbb{Q}^d
\setminus\{(0,\ldots,0)\},
\]
the following first-order sentence belongs to $A_d$:
\begin{equation}
\neg
\left(
f_{q_1}(e_1)+\cdots+f_{q_d}(e_d)=0
\right).
\tag{$A_{15}$}
\end{equation}

In $(A_{15})$, the rational numbers $q_1,\ldots,q_d$ are
metamathematical parameters selecting a particular sentence. Thus
$(A_{15})$ denotes the family of first-order sentences obtained as
$(q_1,\ldots,q_d)$ ranges over $\mathbb{Q}^d$.

The intended structure satisfies these axioms:
\begin{equation}
\mathcal{V}_d\models A_d.
\end{equation}
Consequently,
\begin{equation}
A_d\subseteq T_d.
\end{equation}

The axiom system $A_d$ therefore gives an explicit first-order
description of the principal algebraic and order-theoretic properties
of the intended inheritance structure that are expressible in the
current one-sorted language.

\subsubsection{The Expressibility Limitation of the One-Sorted Language}

There is an important distinction between the explicit axiom system
$A_d$ and the complete theory $T_d$. In the present one-sorted
language, the rational numbers are fixed externally through the family
of function symbols $\{f_q\}_{q\in\mathbb{Q}}$. They do not constitute a
separate sort over which first-order quantification is available.

Consequently, the finite-dimensional spanning condition
\begin{equation}
\forall x\;
\exists q_1,\ldots,q_d\in\mathbb{Q}
\quad
x=
f_{q_1}(e_1)+\cdots+f_{q_d}(e_d)
\end{equation}
is not a single first-order formula of $\mathcal{L}_d$. The displayed
quantifiers range over the externally fixed field $\mathbb{Q}$ rather
than over the one-sorted domain of the structure.

Expressing this condition directly in the present language would
require the infinitary expression
\begin{equation}
\forall x\;
\bigvee_{(q_1,\ldots,q_d)\in\mathbb{Q}^d}
\left(
x=
f_{q_1}(e_1)+\cdots+f_{q_d}(e_d)
\right),
\end{equation}
which is not a first-order formula.

The same observation applies to the full coordinatewise-order
characterization if it is expressed as the existence of nonnegative
rational coordinates for arbitrary elements. Such a condition is part
of the specified intended structure $\mathcal{V}_d$ rather than an
additional first-order axiom of $A_d$ in the present one-sorted
language.

Accordingly, $A_d$ is not claimed to be a categorical first-order
axiomatization of $\mathcal{V}_d$. In particular, we do not claim
\begin{equation}
\operatorname{Mod}(A_d)
=
\{\mathcal{V}_d\}
\quad\text{up to isomorphism}.
\end{equation}

The same distinction applies to $T_d$. Completeness of a first-order
theory and categorical characterization of a structure are distinct
properties. The complete theory $T_d$ contains every first-order
sentence true in $\mathcal{V}_d$, but completeness alone does not imply
that every model of $T_d$ is isomorphic to $\mathcal{V}_d$.

This does not conflict with the canonicity theorem established in
\S\ref{sec:canon}. That theorem concerns the class of structures
satisfying the externally specified fixed-dimensional coordinate
conditions. Within that class, the coordinate map establishes
isomorphism with $\mathbb{Q}^d$. The present first-order construction
addresses a different question: which first-order sentences are true
in the intended structure and which of its expressible properties can
be explicitly axiomatized in the chosen language.

In summary: $A_d$ axiomatizes the expressible algebraic properties of
$\mathcal{V}_d$ in the one-sorted language; it is not categorical, and
$\operatorname{Mod}(A_d)$ properly contains the admissible class. The
canonicity result of \S\ref{sec:canon} concerns the admissible class,
defined externally; the first-order theory $T_d$ is a separate object,
complete by construction but not categorical in $\mathcal{L}_d$.

\subsubsection{Completeness of the First-Order Theory}

\begin{proposition}[Completeness by construction]
\label{prop:comp}
The theory $T_d=\operatorname{Th}_{\mathcal{L}_d}(\mathcal{V}_d)$ is
complete.
\end{proposition}

\begin{proof}
For any $\varphi\in\operatorname{Sent}(\mathcal{L}_d)$, either
$\mathcal{V}_d\models\varphi$ or $\mathcal{V}_d\models\neg\varphi$. By
definition of $T_d$, whichever holds is in $T_d$. Hence for every
$\varphi$, $\varphi\in T_d$ or $\neg\varphi\in T_d$.
\end{proof}

\begin{remark}
This is a property of every theory of the form
$\operatorname{Th}(M)$; it is recorded here not as a novel result but as
a \emph{design certificate}: the formalization leaves no first-order
question about the inheritance structure undecided. It is a statement
about the construction, and is logically independent of both the
canonicity theorem of \S\ref{sec:canon} and the decidability result of
\S\ref{sec:dec}.
\end{remark}

\subsubsection{Structure of the Formal Argument}

The preceding results can be summarized by the following hierarchy:
\begin{equation}
\boxed{
\begin{array}{ccc}
& \mathcal{V}_d=\mathbb{Q}^d & \\[3pt]
\swarrow & & \searrow \\[3pt]
\begin{array}{c}
A_d\subseteq T_d=\operatorname{Th}(\mathcal{V}_d)\\
\text{(complete by construction)}
\end{array}
& &
\begin{array}{c}
\text{finite rational}\\
\text{constraint systems}\\
\downarrow\ \text{Fourier--Motzkin}\\
\text{decidable (existential)}
\end{array}
\end{array}}
\end{equation}

The three principal conclusions are therefore logically distinct:
\begin{align}
&\text{structural canonicity of admissible fixed-dimensional models},
\notag\\
&\text{completeness of }T_d\text{ by construction},
\notag\\
&\text{algorithmic decidability of the inheritance-relevant fragment}.
\notag
\end{align}

The following section develops the third property explicitly.


\section{Algorithmic Decidability of Inheritance Constraint Systems}
\label{sec:dec}

\subsection{Scope of the decidability claim}

The result established in this section is the decidability of the
\emph{inheritance-relevant fragment} of the theory of $\mathcal{V}_d$:
the existential theory over finite quantifier-free systems of rational
linear constraints. This is precisely the fragment generated by the
representation lemma of \S\ref{sec:rep}. Full decidability of $T_d$ is
not claimed and is not needed for the computational-Figh objective:
every question arising from an inheritance configuration is
quantifier-free and linear.

The model-theoretic completeness established above is a semantic
property of the theory $T_d$. The computational problem addressed in
this section is different. Given a concrete inheritance configuration,
the legal conditions produce a finite system of rational linear
equalities and inequalities. The objective is to provide a terminating
procedure that determines whether the system is consistent and, when
the constraints determine a unique allocation, computes that
allocation exactly.

\subsubsection{Input Representation}

Let an inheritance configuration contain $n$ rational unknowns
\[
x_1,\ldots,x_n.
\]
After normalization of the estate and translation of the legal
conditions into algebraic form, the configuration is represented by a
finite system
\begin{equation}
\mathcal{S}
=
\left\{
a_j^{\mathsf T}x\leq b_j
:
1\leq j\leq m
\right\},
\qquad
x=(x_1,\ldots,x_n)\in\mathbb{Q}^n,
\end{equation}
where
\[
a_j\in\mathbb{Q}^n,
\qquad
b_j\in\mathbb{Q}.
\]

Every equality
\[
a^{\mathsf T}x=b
\]
is replaced by the two inequalities
\[
a^{\mathsf T}x\leq b
\qquad\text{and}\qquad
-a^{\mathsf T}x\leq-b.
\]

Consequently, by Lemma~\ref{lem:rep}, every finite inheritance
configuration can be reduced to a finite system of rational linear
inequalities.

The computational problem is therefore:

\begin{equation}
\boxed{
\text{Given }\mathcal{S},
\text{ determine whether }
\exists x\in\mathbb{Q}^n
\text{ satisfying every constraint in }\mathcal{S}.
}
\end{equation}

\subsubsection{Fourier--Motzkin Elimination}

Fourier--Motzkin elimination removes one variable at a time while
preserving satisfiability over the rational numbers.

Suppose the current system contains the variable $x_n$. Every
inequality can be written in the form
\begin{equation}
a_i x_n+c_i(x_1,\ldots,x_{n-1})\leq b_i.
\end{equation}

Partition the inequalities into three classes according to the
coefficient $a_i$:
\begin{align}
\mathcal{U}_n
&=
\{i:a_i>0\},
\\
\mathcal{L}_n
&=
\{i:a_i<0\},
\\
\mathcal{Z}_n
&=
\{i:a_i=0\}.
\end{align}

For $i\in\mathcal{U}_n$, division by the positive rational $a_i$
gives an upper bound
\begin{equation}
x_n
\leq
U_i(x_1,\ldots,x_{n-1}).
\end{equation}

For $j\in\mathcal{L}_n$, division by the negative rational $a_j$
reverses the inequality and gives a lower bound
\begin{equation}
x_n
\geq
L_j(x_1,\ldots,x_{n-1}).
\end{equation}

The variable $x_n$ can exist exactly when every lower bound is
compatible with every upper bound. Therefore Fourier--Motzkin replaces
all inequalities involving $x_n$ by
\begin{equation}
L_j(x_1,\ldots,x_{n-1})
\leq
U_i(x_1,\ldots,x_{n-1}),
\qquad
j\in\mathcal{L}_n,\quad i\in\mathcal{U}_n,
\end{equation}
while retaining all inequalities in $\mathcal{Z}_n$.

The resulting system, denoted
\[
\operatorname{FM}_{x_n}(\mathcal{S}),
\]
contains only the variables
\[
x_1,\ldots,x_{n-1}.
\]

The fundamental elimination property is
\begin{equation}
\mathcal{S}
\text{ is satisfiable}
\quad\Longleftrightarrow\quad
\operatorname{FM}_{x_n}(\mathcal{S})
\text{ is satisfiable}.
\end{equation}

Thus elimination of one variable preserves satisfiability.

The equivalence holds over $\mathbb{Q}$, not merely over
$\mathbb{R}$, because the bounds are rational affine functions of the
remaining variables: at any rational point of the projected system, the
maximum lower bound and minimum upper bound are rational, and any
rational value between them (e.g.\ the midpoint) extends to a rational
solution of the pre-elimination system. This rationality argument is
what licenses the exact rational reconstruction of \S\ref{sec:dec}.

\subsubsection{The Complete Decision Procedure}

Starting with
\[
\mathcal{S}_n=\mathcal{S},
\]
apply Fourier--Motzkin elimination successively:
\begin{equation}
\mathcal{S}_n
\longrightarrow
\mathcal{S}_{n-1}
\longrightarrow
\cdots
\longrightarrow
\mathcal{S}_1
\longrightarrow
\mathcal{S}_0,
\end{equation}
where $\mathcal{S}_k$ contains only the variables
$x_1,\ldots,x_k$.

At the final stage $\mathcal{S}_0$ contains only rational inequalities
of the form
\[
c\leq 0,
\qquad
c\in\mathbb{Q}.
\]

If the final system contains a contradiction such as
\begin{equation}
c>0
\end{equation}
in a constraint equivalent to
\[
c\leq0,
\]
then the original inheritance system is inconsistent.

If no contradiction occurs, the original system is satisfiable.

Thus the decision criterion is
\begin{equation}
\boxed{
\mathcal{S}\text{ is satisfiable}
\quad\Longleftrightarrow\quad
\mathcal{S}_0\text{ contains no contradictory rational inequality}.
}
\end{equation}

Because each elimination step removes one variable and the original
system contains only finitely many variables and inequalities, the
procedure terminates after finitely many elimination stages.

\subsubsection{Correctness of the Decision Procedure}

\begin{theorem}
For every finite system $\mathcal{S}$ of linear equalities and
inequalities over $\mathbb{Q}$, the Fourier--Motzkin procedure
terminates and correctly determines whether $\mathcal{S}$ has a
rational solution.
\end{theorem}

\begin{proof}
Equalities are first replaced by pairs of inequalities, so it is
sufficient to consider systems of linear inequalities.

At each elimination step, the current variable is partitioned into
upper-bound, lower-bound, and zero-coefficient constraints. A value of
the eliminated variable can satisfy the original constraints exactly
when every lower bound is at most every upper bound. The pairwise
combination of lower and upper bounds therefore produces a system
satisfied by exactly the projections of the original solution set onto
the remaining variables.

Hence
\[
\mathcal{S}_k
\text{ is satisfiable}
\quad\Longleftrightarrow\quad
\mathcal{S}_{k-1}
\text{ is satisfiable}
\]
at every elimination step.

After all $n$ variables have been eliminated, the resulting system
contains only rational constant inequalities. Such a system is
decidable by direct comparison of rational numbers. If a contradiction
occurs, the original system has no solution. If no contradiction
occurs, the successive projection equivalences imply that the original
system has a solution.

Since exactly one variable is eliminated at each stage and the initial
system is finite, the procedure terminates after at most $n$
elimination stages.
\end{proof}

\subsubsection{Reconstruction of a Solution}

The elimination procedure does not merely decide consistency. When a
solution is required, the eliminated variables can be reconstructed
in reverse order.

Suppose the final system $\mathcal{S}_0$ is satisfiable. Choose a
solution of the remaining system and return to the preceding
elimination stage.

At that stage, the eliminated variable $x_k$ is constrained by the
finite collection of rational lower and upper bounds that remain after
the previously eliminated variables have been fixed. If both lower and
upper bounds are present, the Fourier--Motzkin constraints guarantee
that every lower bound does not exceed every upper bound. If one side
is unbounded, any rational value satisfying the bounds that are
present may be chosen. Thus a rational value of $x_k$ can always be
selected whenever the projected system is satisfiable.

The procedure then substitutes the selected value of $x_k$ into the
constraints from the preceding elimination stage and continues
backwards. Repeating this process reconstructs a rational solution
\[
(x_1,\ldots,x_n)
\]
of the original system.

Therefore the computational procedure has two modes of output:
\begin{equation}
\begin{array}{ll}
\text{inconsistent system}
&
\longrightarrow
\text{no admissible allocation},
\\[1mm]
\text{consistent system}
&
\longrightarrow
\text{an exact rational allocation}.
\end{array}
\end{equation}

When the legal constraints determine a unique allocation, the
reconstruction necessarily yields the same rational value for every
solution.

\subsubsection{Determining Uniqueness}

After the satisfiability of the system has been established, the
procedure can determine whether the resulting allocation is unique.
Let
\[
\mathcal{S}
\]
be the original finite system of rational linear constraints, and let
\[
P
=
\left\{
x\in\mathbb{Q}^{n}
:
\mathcal{S}(x)
\right\}
\]
denote its rational feasible set.

For each coordinate $x_i$, the procedure computes the projection of
the feasible set onto the $x_i$-coordinate by applying
Fourier--Motzkin elimination to all variables other than $x_i$.
The resulting one-variable system determines the feasible lower and
upper bounds
\[
\ell_i=\inf_{x\in P}x_i
\qquad\text{and}\qquad
u_i=\sup_{x\in P}x_i.
\]

If $\ell_i=-\infty$ or $u_i=+\infty$, the coordinate is unbounded and
the allocation is not unique. If both bounds are finite, the feasible
values of $x_i$ form the closed interval $[\ell_i,u_i]$ (closed because
rational linear constraints define a closed set in the order topology,
and the bounds are attained at rational points by the rationality
argument above). The allocation is uniquely determined if and only if
every coordinate has a singleton feasible interval. Equivalently,
\begin{equation}
\ell_i=u_i
\qquad
\text{for every }i=1,\ldots,n.
\end{equation}

In that case, the unique allocation is
\[
x^*=(\ell_1,\ldots,\ell_n)
    =(u_1,\ldots,u_n).
\]

If, on the other hand, there exists some coordinate $x_i$ for which
\[
\ell_i<u_i,
\]
then at least two distinct feasible allocations exist, since that
coordinate can take different feasible values. Hence the original
system does not determine a unique allocation.

Accordingly, the complete decision procedure distinguishes the
following cases:
\begin{enumerate}
    \item if $\mathcal{S}$ is unsatisfiable, the inheritance
    configuration has no admissible allocation;

    \item if $\mathcal{S}$ is satisfiable and
    \[
    \ell_i=u_i
    \quad\text{for every }i,
    \]
    then $\mathcal{S}$ determines a unique rational allocation,
    which is computed exactly;

    \item if $\mathcal{S}$ is satisfiable but
    \[
    \ell_i<u_i
    \]
    for at least one coordinate $x_i$, then the constraints admit
    multiple allocations.
\end{enumerate}

All bounds are obtained from finite systems of rational linear
constraints, so each required elimination terminates after finitely
many steps. Consequently, uniqueness of the inheritance allocation is
itself decidable by the same terminating Fourier--Motzkin framework.

\subsection{Complexity}

Fourier--Motzkin elimination is terminating and exact, but its
worst-case complexity is doubly exponential in the number of variables,
since the number of constraints can square at each elimination step
\cite{Dantzig1973,Schrijver1986,Chandru1993}. For the fixed, small
number of claimants in typical inheritance configurations this is
acceptable; for larger configurations, standard linear-programming
methods (simplex, interior point) or polyhedral libraries may be
substituted without affecting the correctness argument, since they
decide the same fragment. The choice of Fourier--Motzkin in this paper
is for transparency of the correctness proof, not for asymptotic
efficiency.

\subsubsection{Application to the Historical Bequest}

Consider again the system
\begin{align}
2x+w&=1,
\\
w&=\frac13-\frac14x,
\\
w&\leq\frac13.
\end{align}

Replacing the equalities by pairs of inequalities gives
\begin{align}
w&\leq1-2x,
&
w&\geq1-2x,
\\
w&\leq\frac13-\frac14x,
&
w&\geq\frac13-\frac14x,
\\
w&\leq\frac13.
\end{align}

Eliminating $w$ combines every lower bound with every upper bound. The
two nontrivial projected constraints are
\begin{align}
1-2x &\leq \tfrac13-\tfrac14x
&&\Longrightarrow&
x &\geq \tfrac{8}{21},
\\
\tfrac13-\tfrac14x &\leq 1-2x
&&\Longrightarrow&
x &\leq \tfrac{8}{21}.
\end{align}
Together they give $x=\tfrac{8}{21}$ directly, without appealing to
the retained equality. The legal cap $w\leq\tfrac13$ has already been
used in forming the first projection.

Substitution gives
\begin{equation}
w
=
\frac13-\frac14\left(\frac{8}{21}\right)
=
\frac5{21}.
\end{equation}

Thus the algorithm returns the exact rational allocation
\begin{equation}
\boxed{
x=\frac8{21},
\qquad
w=\frac5{21}.
}
\end{equation}

The result also satisfies the legal upper bound:
\begin{equation}
\frac5{21}\leq\frac13.
\end{equation}

\subsection{Application to an \textit{`Awl} configuration}

Consider a deceased woman leaving a husband and two sisters. The
prescribed shares are $1/2$ for the husband and $2/3$ for the two
sisters, summing to $7/6>1$. Introduce unknowns $x$ (husband),
$y$ (each sister), and a scaling variable $s$, subject to
\begin{align}
x+2y &= 1, &&\text{(Conservation)},
\\
2x - s &= 0, &&\text{(Proportional share of husband)},
\\
3y - s &= 0, &&\text{(Proportional share of each sister)},
\\
s &\leq 1, &&\text{(Estate cap)},
\\
s &\geq 0. &&
\end{align}
The first three equalities are each replaced by a pair of inequalities.
Eliminating $s$ yields the projected constraints
\begin{equation}
2x \leq 3y, \qquad 3y \leq 2x, \qquad 2x \leq 1, \qquad 3y \leq 1.
\end{equation}
The first two combine to $2x=3y$; substituting into $x+2y=1$ gives
\[
x+2y=1,\qquad y=\tfrac{2x}{3}
\quad\Longrightarrow\quad
\tfrac{7x}{3}=1
\quad\Longrightarrow\quad
x=\tfrac{3}{7},\quad y=\tfrac{2}{7}.
\]
Substitution gives $s=6/7$, which satisfies $s\leq 1$.

Thus the algorithm returns
\begin{equation}
\boxed{
x=\tfrac{3}{7},\qquad y=\tfrac{2}{7},\qquad s=\tfrac{6}{7}.
}
\end{equation}

The procedure returns the same allocation as the classical
\textit{`awl} computation. This is expected: \textit{`awl} solves
$\sum s_i X=1$, which is the single-variable case of the general linear
system. The contribution of the algebraic representation is not a
different answer but an explicit structure in which the answer is
derived by a uniform procedure and its correctness is mechanically
checkable.

\subsubsection{Termination and Computational Character}

The decision procedure terminates because the number of variables in a
fixed inheritance configuration is finite and each Fourier--Motzkin
stage eliminates at least one variable. No search through an
unbounded space of numerical assignments is required. All arithmetic
operations are performed over the rational numbers, and therefore the
resulting allocations are exact rather than numerical approximations.

The procedure can consequently be summarized as
\begin{equation}
\boxed{
\begin{aligned}
\text{inheritance configuration}
&\longrightarrow
\text{finite rational constraint system}
\\
&\longrightarrow
\text{equality normalization}
\\
&\longrightarrow
\text{successive Fourier--Motzkin elimination}
\\
&\longrightarrow
\text{consistency decision}
\\
&\longrightarrow
\text{exact solution reconstruction}.
\end{aligned}}
\end{equation}

This establishes algorithmic decidability for every fixed inheritance
configuration represented by a finite system of rational linear
constraints, as licensed by the representation lemma of
\S\ref{sec:rep}. The result is independent of the completeness of
$T_d$ by construction: the former is a property of a terminating
computational procedure, whereas the latter is a property of the
first-order theory associated with the intended structure.


\section{Methodological interpretation: the Fiqh-domain analogy}
\label{sec:fiqh}

This section states a methodological analogy, not a theorem of the
present paper. The notion of domain-relative computational completeness
is defined and established in the companion framework \cite{CompFiqh};
the present section records how the inheritance construction fits that
pattern.

The significance of the preceding construction for the broader Fiqh
framework is methodological. Al-Khw\={a}rizm\={i}'s inheritance
treatment provides a general computational method for a specified
class of jurisprudential problems. The formal machinery used in the
present framework is different, but the computational objective is
analogous.

For a closed Fiqh chapter $CC$, the
\textit{al-sabr wa-al-taqs\={i}m} procedure exhaustively generates the
specified case space. The resulting Horn Fiqh axiom set $A_{CC}$ is
complete with respect to that specified finite case space, as
established by the completeness results developed elsewhere in the
formal framework \cite{CompFiqh}.

If the Horn deductive engine is complete for the rule class specified
by the framework, then every ruling represented in $A_{CC}$ is
derivable by the engine. The resulting computational chain is
\begin{equation}
\text{complete case generation}
\longrightarrow
\text{complete axiom set }A_{CC}
\longrightarrow
\text{complete Horn deduction}.
\end{equation}

The analogy holds under three conditions:
\begin{enumerate}
\item[(C1)] the case space of the closed Fiqh chapter is finite and
exhaustively generated by \textit{al-\d{s}abr wa-al-taqs\={i}m};
\item[(C2)] the resulting axiom set $A_{CC}$ is Horn and complete with
respect to that case space \cite{CompFiqh};
\item[(C3)] the deductive engine is complete for the specified Horn
rule class \cite{CompFiqh}.
\end{enumerate}
Under (C1)--(C3), domain-relative computational completeness holds for
the chapter.

Relative to the specified ontology, question space, axioms, and
deductive rule class, this yields
\emph{domain-relative computational completeness}.

The analogy with Al-Khw\={a}rizm\={i}'s inheritance algebra is
methodological rather than a claim that the historical mathematical
machinery was identical to the modern formal machinery. In
Al-Khw\={a}rizm\={i}'s treatment, a delimited class of inheritance
problems is represented through algebraic relations over rational
quantities and resolved by a general computational procedure. In the
present framework, a delimited Fiqh domain is represented through
exhaustive case generation, Horn-clause axiomatization, and a
terminating deductive engine.

In both cases, the objective is the computational treatment of an
entire specified class of jurisprudential problems.

The analogy is methodological, not historical: it does not claim that
Al-Khw\={a}rizm\={i}'s procedure was a Horn deduction, nor that the
Fiqh framework reduces to linear algebra. What is shared is the
pattern: a delimited class of jurisprudential problems, an explicit
finite representation, and a terminating procedure that resolves every
instance in the class.


\section{Conclusion}

The formalization developed here provides a modern mathematical
interpretation of the computational character of Al-Khw\={a}rizm\={i}'s
inheritance algebra. A specified inheritance problem can be represented
as a finite system of rational linear constraints and solved exactly
by a terminating procedure. For a fixed number of claimant coordinates,
the associated mathematical structure is represented by $\mathbb{Q}^d$,
and every admissible realization of that structure is isomorphic to the
intended model.

The complete first-order theory of the intended model is complete by
construction, as a stated design property of the formalization. The
inheritance-relevant existential fragment generated by the
representation lemma is algorithmically decidable by Fourier--Motzkin
elimination. These are separate formal properties, and the paper does
not claim full decidability of $T_d$.

The analysis of the Qur'anic inheritance verses further establishes an
important epistemological distinction between linguistic description,
semantic interpretation, and formal structure. Ambiguity in a natural
language expression does not by itself imply inconsistency in the formal
system represented by that expression. A contradiction obtained under
one selected semantic interpretation establishes, at most, a failure of
that interpretation; it does not establish unsatisfiability of the
underlying axioms unless the interpretation is shown to be uniquely
compelled by the linguistic evidence. This distinction has an important
internal analogue in the Qur'an's own epistemological account of its
text, particularly its distinction between \textit{\={a}y\={a}t
mu\d{h}kam\={a}t}, described as the foundation of the Book, and
\textit{\={a}y\={a}t mutash\={a}bih\={a}t}, whose interpretation is not
equally immediate or determinate (Qur'an 3:7). In the context of
Computational Fiqh, this yields a methodological principle: the
transition from natural-language jurisprudential statements to formal
semantics should be made explicit, so that alternative interpretations
can be formally represented and their resulting structures evaluated
for consistency.

The contribution of the algebraic representation is representational
rather than exclusive: once the inheritance structure is represented
by rational linear coordinates, the representation is canonical within
the admissible class, and every question arising from an inheritance
configuration reduces to an existential question over a finite rational
linear constraint system that is decidable exactly. This explains why
the pre-algebraic and algebraic procedures agree on allocations, and
why the algebraic representation makes the structure of the problem
explicit in a way that the pre-algebraic procedure does not.

The broader significance is methodological. A delimited class of
jurisprudential problems can be represented by an explicit formal
structure together with a general computational procedure for resolving
its instances. The broader Fiqh framework applies the same principle
through exhaustive case generation, complete axiomatization, and
terminating Horn-clause deduction.

The resulting notion of \emph{domain-relative computational completeness}
does not assert universal formalization. It asserts that, for a
specified closed Fiqh domain and relative to its explicitly defined
ontology, question space, axioms, and deductive mechanism, the relevant
cases can be completely represented and computationally resolved.



\begin{thebibliography}{99}

\bibitem{Rashed1994}
Rashed, R. (1994). \textit{The Development of Arabic Mathematics: Between Arithmetic and Algebra}. Springer-Verlag.

\bibitem{Berggren1986}
Berggren, J. L. (1986). \textit{Episodes in the Mathematics of Medieval Islam}. Springer-Verlag.

\bibitem{Said1978}
Said, E. W. (1978). \textit{Orientalism}. Pantheon Books.

\bibitem{Tibawi1964}
Tibawi, A. L. (1964). English-speaking Orientalists: A Critique of their Approach to Islam and Arab Nationalism. \textit{Islamic Quarterly}, 8(1/2), 25--45.

\bibitem{Crone1987}
Crone, P. (1987). \textit{Meccan Trade and the Rise of Islam}. Princeton University Press.

\bibitem{Burton1990}
Burton, J. (1990). \textit{The Sources of Islamic Law: Islamic Theories of Abrogation}. Edinburgh University Press.

\bibitem{Powers2011}
Powers, D. S. (2011). \textit{Muhammad is Not the Father of any of Your Men: The Making of the Last Prophet}. University of Pennsylvania Press.

\bibitem{Popper1963}
Popper, K. (1963). \textit{Conjectures and Refutations: The Growth of Scientific Knowledge}. Routledge.

\bibitem{Hallaq2013}
Hallaq, W. B. (2013). \textit{The Impossible State: Islam, Politics, and Modernity's Moral Predicament}. Columbia University Press.

\bibitem{Deroche2014}
D\'{e}roche, F. (2014). \textit{Qur'ans of the Umayyads: A First Overview}. Brill.

\bibitem{Sadeghi2012}
Sadeghi, B., \& Goudarzi, M. (2012). \d{S}an`\={a}' 1 and the Origins of the Qur'\={a}n. \textit{Der Islam}, 87(1-2), 1--129.

\bibitem{Oaks2014}
Oaks, J. (2014). Al-Khw\={a}rizm\={i} and the New Genre of Algebra. \textit{Historia Mathematica}, 41(4), 459--473.

\bibitem{Dantzig1973}
Dantzig, G.~B., \& Eaves, B.~C. (1973). Fourier--Motzkin elimination and its dual. \textit{Journal of Combinatorial Theory, Series A}, 14(3), 288--297.

\bibitem{Schrijver1986}
Schrijver, A. (1986). \textit{Theory of Linear and Integer Programming}. Wiley.

\bibitem{Chandru1993}
Chandru, V. (1993). Variable elimination in linear constraints. \textit{The Computer Journal}, 36(5), 463--472.

\bibitem{Marker2002}
Marker, D. (2002). \textit{Model Theory: An Introduction}. Springer.

\bibitem{CompFiqh}
Elnaserledinellah Mahmoud Abdelwahab. Theoretical Foundations of Computational Islamic Jurisprudence (Fiqh). 2026. ⟨hal-05739614⟩

\end{thebibliography}
\end{document}